\documentclass[runningheads,envcountsame]{llncs}
\usepackage[utf8]{inputenc}
\usepackage{xspace}
\usepackage{url}
\usepackage{listings}
\usepackage{shuffle}
\usepackage{amsmath}
\usepackage{stackrel}
\usepackage{amssymb,amsfonts}
\usepackage{enumerate}
\usepackage{tikz}
\usepackage[english]{babel}

\usetikzlibrary{shapes,shapes.geometric,arrows,fit,calc,positioning,automata,decorations.pathmorphing}

\def\RR{\mathbb{R}}

\newcommand{\lang}{\mathcal{L}}

\newcommand{\DFAs}{\textsf{DFAs}\xspace}
\newcommand{\NFA}{\textsf{NFA}\xspace}
\newcommand{\NFAs}{\textsf{NFAs}\xspace}
\newcommand{\nfa}{\mathcal{A}}

\newcommand{\nodot}{ }
\newcommand{\re}{\textsf{RE}\xspace}
\newcommand{\res}{\textsf{REs}\xspace}

\newcommand{\Ash}{\Sigma}
\newcommand{\Amat}{\mathsf{A}}
\newcommand{\Mmat}{\mathsf{M}}
\newcommand{\Emat}{\mathsf{E}}
\newcommand{\Cmat}{\mathsf{C}}

\newcommand{\tsh}{\mathsf{T}_\shuffle}

\newcommand{\apd}{\mathcal{A}_{pd}}
\newcommand{\apos}{\mathcal{A}_{pos}}

\newcommand{\eqsynt}{=}
\newcommand{\eqsem}{=}

\newcommand{\fado}{\textsf{FAdo}\xspace}

\title{Partial Derivative Automaton \\ for Regular Expressions with Shuffle}

\author{Sabine Broda \and 
António Machiavelo \and Nelma Moreira \and Rogério Reis
}

\institute{CMUP \& DCC, Faculdade de Ciências da Universidade do Porto,
Portugal\\
Rua do Campo Alegre, 4169-007 Porto, Portugal\\
\email{sbb@dcc.fc.up.pt,ajmachia@fc.up.pt,\{nam,rvr\}@dcc.fc.up.pt}}

\begin{document}

\maketitle

\begin{abstract}
We generalize the partial derivative automaton to regular expressions with shuffle and study its size in the worst and in the average case. The number of states of the partial derivative automata is in the worst case at most $2^m$, where $m$ is the number of letters in the expression, while asymptotically and on average it is no more than~$(\frac43)^m$.
\end{abstract}

\section{Introduction}
\label{sec:intro}
The shuffle (or interleaving) operation is closed for regular
languages, and extended regular expressions with shuffle can be much
more succinct than the equivalent ones with disjunction,
concatenation, and star operators. For the shuffle operation, Mayer
and Stockmeyer~\cite{mayer94:_word_probl_this_time_with_inter} studied
the computational complexity of membership and inequivalence
problems. Inequivalence is exponential time complete, and membership
is NP-complete for some classes of regular languages. In particular,
they showed that for regular expressions (\res) with shuffle, of size
$n$, an equivalent nondeterministic finite automaton (\NFA) needs at
most $2^n$ states, and presented a family of \res with shuffle, of
size ${\cal O}(n)$, for which the correspondent \NFAs have at least
$2^n$ states.  Gelade~\cite{gelade10:_succin_of_regul_expres_with},
and Gruber and
Holzer~\cite{gruber08:_finit_autom_digrap_connec_and,gruber10:_descr_and_algor_compl_of_regul_languag}
showed that there exists a double exponential trade-off in the
translation from \res with shuffle to stantard \res. Gelade also gave
a tight double exponential upper bound for the translation of \res
with shuffle to \DFAs.  Recently, conversions of shuffle expressions
to finite automata were presented by
Estrade~\cite{estrade06:_explic_paral_regul_expres} and Kumar and
Verma~\cite{kumar14:_novel_algor_for_conver_of}. In the latter paper
the authors give an algorithm for the construction of an
$\varepsilon$-free \NFA based on a classic Glushkov construction, and
the authors claim that the size of the resulting automaton is at most
$2^{m+1}$, where $m$ is the number of letters that occur in the \re
with shuffle.

In this paper we present a conversion of \res with shuffle to
$\varepsilon$-free \NFAs, by generalizing the partial derivative
construction for standard
\res~\cite{antimirov96:_partial_deriv_regul_expres_finit_autom_const,b.g.mirkin66:_algor_for_const_base_in}.
For standard \res, the partial derivative automaton ($\apd$) is a
quotient of the Glushkov automaton ($\apos$) and Broda \emph{et
  al.}~\cite{broda11:_averag_state_compl_of_partial_deriv_autom,broda12:_averag_size_of_glush_and}
showed that, asymptotically and on average, the size of $\apd$ is half
the size of $\apos$.  In the case of \res with shuffle we show that
the number of states of the partial derivative automaton is in the
worst-case $2^m$ (with $m$ as before) and an upper bound for the
average size is $(\frac43)^m$.

This paper is organized as follows. In the next section we review the
shuffle operation and regular expressions with shuffle. In
Section~\ref{sec:autseq} we consider equation systems, for languages
and expressions, associated to nondeterministic finite automata and
define a solution for a system of equations for a shuffle
expression. An alternative and equivalent construction, denoted by
$\apd$, is given in Section~\ref{sec:pd} using the notion of partial
derivative. An upper bound for the average number of states of $\apd$
using the framework of analytic combinatorics is given in
Section~\ref{sec:average}. We conclude in Section~\ref{sec:conclusion}
with some considerations about how to improve the presented upper
bound and related future work.

%

\section{Regular Expressions with Shuffle}
\label{sec:reshu}
Given an alphabet $\Sigma$, the shuffle of two words in $\Sigma^\star$ is a finite set of words defined inductively  as follows, for $x,y\in \Sigma^\star$ and $a,b\in \Sigma$
\begin{eqnarray*}
  x \shuffle \varepsilon = \varepsilon \shuffle x & = & \{x\} \label{eq:grammar}
\\
ax \shuffle by &=& \{\; az \mid z \in x \shuffle by\;\} \cup \{\; bz \mid z \in ax \shuffle y\;\}.
\end{eqnarray*}

This definition is extended to sets of words, i.e. languages, in the natural way:
$$L_1 \shuffle L_2 = \{\; x \shuffle y \mid x \in L_1, y \in L_2\;\}.$$

It is well known that if two languages $L_1, L_2\subseteq \Sigma^\star$ are regular then $L_\varepsilon \shuffle L_2$ is regular. One can extent regular expressions to include the $\shuffle$ operator.
Given an alphabet $\Ash$, we denote by $\tsh$ the set containing $\emptyset$ 
plus all terms finitely generated from $\Ash
  \cup \{\varepsilon\}$ and operators $+,\cdot,\shuffle,{}^*$, that is, the expressions $\tau$ generated by
the grammar
\begin{eqnarray} 
\tau & \rightarrow & \emptyset \mid \alpha \\
\alpha  & \rightarrow & \varepsilon \mid a \mid
\alpha + \alpha \mid \alpha \cdot \alpha \mid \alpha \shuffle \alpha
\mid \alpha ^\star \quad (a\in \Ash)  \label{gram:shuffle}.
\end{eqnarray} 

As usual, the (regular) language $\lang (\tau)$ represented by an
expression $\tau \in \tsh$ is inductively defined as follows:
$\lang(\emptyset) =\emptyset$, $\lang(\varepsilon) =\{\varepsilon \}$,
$\lang(a) = \{a\}$ for $a\in \Sigma$, $\lang(\alpha^\star) =
\lang(\alpha)^\star$, $\lang(\alpha + \beta) = \lang(\alpha) \cup
\lang(\beta)$, $\lang(\alpha \beta) =\lang(\alpha)\lang(\beta)$, and
$\lang(\alpha \shuffle \beta) = \lang(\alpha) \shuffle \lang(\beta).$
We say that two expressions $\tau_1, \tau_2 \in \tsh$ are equivalent, and write $\tau_1 \eqsem \tau_2$, if $\lang(\tau_1) = \lang(\tau_2)$.

\begin{example}
\label{ex:lang}
Consider $\alpha_n = a_1 \shuffle \cdots \shuffle a_n$, where $n \geq 1$, $a_i \not= a_j$ for $1  \leq i \not= j \leq n$. Then,
$$\lang(\alpha_n) = \{\; a_{i_1} \cdots a_{i_n} \mid i_1, \ldots ,i_n \mbox{ is a permutation of } 1,\ldots,n \}.$$
\end{example}

The shuffle operator $\shuffle$ is commutative, associative, and distributes over $+$. It also follows that for all $a,b\in  \Sigma$ and $\tau_1,\tau_2\in \tsh$, $$a\tau_1\shuffle b\tau_2=a(\tau_1\shuffle b\tau_2)+b(a\tau_1\shuffle \tau_2).$$

 Given a language $L$, we define $\varepsilon(\tau)=\varepsilon(\lang(\tau))$, where, $\varepsilon(L)=\varepsilon$ if $\varepsilon \in L$ and
$\varepsilon(L)=\emptyset$ otherwise.
 A recursive definition of
$\varepsilon : \tsh \longrightarrow \{\emptyset,\varepsilon\}$ is given by the
following: $ \varepsilon(a) =
\varepsilon(\emptyset) =\emptyset$, 
$\varepsilon(\varepsilon) =
\varepsilon(\alpha^\ast) =\varepsilon$,
$\varepsilon(\alpha+\beta) = \varepsilon(\alpha) + \varepsilon(\beta)$,
$\varepsilon(\alpha \beta) =\varepsilon(\alpha) \varepsilon(\beta)$, and
$\varepsilon(\alpha \shuffle \beta) = \varepsilon(\alpha) 
\varepsilon(\beta).$


\section{Automata and Systems of Equations}
\label{sec:autseq}
We first recall the definition of a \NFA as a tuple $\nfa = \langle S, \Sigma , S_0, \delta, F \rangle$, where $S$ is a finite set of states, $\Sigma$ is a finite alphabet, $S_0 \subseteq S$ a set of initial states, $\delta: S \times \Sigma \longrightarrow {\cal P}(S)$ the transition function, and $F \subseteq S$ a set of final states.
The extension of $\delta$ to sets of states and words is defined by $\delta(X,\varepsilon) = X$ and $\delta(X,a x) = \displaystyle{\delta(\cup_{s \in X}\delta(s,a),x)}$.  A word $x \in\Sigma^\ast$ is accepted by $\nfa$ if and only if $\delta(S_0,x) \cap F \neq \emptyset$.  The \emph{language of $\nfa$} is the set of words accepted by $\nfa$ and denoted by $\lang(\nfa)$. The \emph{right language of a state} $s$, denoted by $\lang_s$, is the language accepted by $\nfa$ if we take $S_0 =\{s\}$. The class of languages accepted by all the $\NFAs$ is precisely  the set of regular languages. 

It is well known that, for each $n$-state \NFA $\nfa$, over
$\Sigma=\{a_1, \ldots, a_k \}$, having right languages $\lang_1,
\ldots ,\lang_n$, it is possible to associate a system of linear
language equations
  \begin{eqnarray*}
\label{eq:system}
\lang_i &=& a_1  \lang_{1i} \cup\cdots \cup a_k  \lang_{ik} \cup\varepsilon(\lang_i), \quad i\in [1,n]
  \end{eqnarray*}
where each $\lang_{ij}$ is a (possibly empty) union of elements in $\{ \lang_1, \ldots, \lang_n\}$, and $\lang_1=\lang(\nfa)$.

In the same way, it is possible to associate to each regular
expression a system of equations on expressions. We here extend this
notion to regular expressions with shuffle.

\begin{definition}
  Consider $\Sigma= \{a_1,\ldots,a_k\}$ and $\alpha_0 \in
  \tsh$. A \emph{support} of $\alpha_0$ is a set
  $\{\alpha_1,\ldots,\alpha_n\}$ that satisfies a system
  of equations
  \begin{eqnarray}
\label{eq:prebase}
\alpha_i &\eqsem& a_1 \alpha_{1i} + \cdots + a_k \alpha_{ki} + \varepsilon(\alpha_i), \quad i\in [0,n]
  \end{eqnarray}
 for some $\alpha_{1i}, \ldots, \alpha_{ki}$, each one a (possibly
  empty) sum of elements in $\{\alpha_1,\ldots,\alpha_n\}$.
  In this case
  $\{\alpha_0, \alpha_1,\ldots,\alpha_n\}$ is called a \emph{prebase}
  of $\alpha_0$.
\end{definition}

It is clear from what was just said above, that the existence of a
support of $\alpha$ implies the existence of an \NFA that accepts the
language determined by $\alpha$.  

Note that the system of equations~\eqref{eq:prebase} can be written in
matrix form $\Amat_\alpha \eqsem\Cmat \cdot \Mmat_\alpha +
\Emat_\alpha,$ where $\Mmat_\alpha$ is the $k\times (n+1)$ matrix with
entries $\alpha_{ij}$, and $\Amat_\alpha$, $\Cmat$ and $\Emat_\alpha$
denote respectively the following three matrices,
$$\Amat_\alpha =\begin{bmatrix} \alpha_0  & \cdots &
  \alpha_n \end{bmatrix}, \quad \Cmat =\begin{bmatrix} a_1 & \cdots
  & a_k \end{bmatrix}, \quad \mbox{and} \quad \Emat_\alpha
=\begin{bmatrix} \varepsilon(\alpha_0) & \cdots &
  \varepsilon(\alpha_n) \end{bmatrix},$$ 
where, $\Cmat \cdot \Mmat_\alpha$ denotes the matrix obtained from $\Cmat$ and
$\Mmat_\alpha$ applying the standard rules of matrix multiplication,
but replacing the multiplication by the concatenation. This notation
will be use below.

A support for an expression $\alpha \in \tsh$ can be computed using
the function $\pi : \tsh \longrightarrow {\cal P}(\tsh)$ recursively
given by the following.

\begin{definition}
\label{def:pi}
Given $\tau \in \tsh$, the set
    $\pi(\tau)$ is
  inductively defined by,
  \begin{center}$ 
\begin{array}[t]{rcl}
    \pi(\emptyset)&=&\pi(\varepsilon) = \emptyset \\
    \pi(a) &=& \{ \varepsilon\} \quad (a \in \Ash)\\
     \pi(\alpha^*) &=& \pi(\alpha)\nodot  \alpha^*
      \end{array}\qquad \qquad 
     \begin{array}[t]{rcl}
      \pi(\alpha+\beta) &=& \pi(\alpha) \cup \pi(\beta) \\
    \pi(\alpha \beta) &=& \pi(\alpha) \nodot \beta \cup \pi(\beta)
    \\
    \pi(\alpha \shuffle \beta) &=& \pi(\alpha) \shuffle \pi(\beta) \\
&\cup&
    \pi(\alpha) \shuffle \{ \beta \} \cup \{ \alpha \} \shuffle \pi(\beta),
  \end{array}$
  \end{center}
where, given $S,T \subseteq \tsh$ and $\beta\in \tsh\setminus\{\emptyset, \varepsilon\}$, $S \beta = \{\; \alpha\beta \mid \alpha \in S\;\}$ and $S \shuffle T = \{\; \alpha \shuffle \beta \mid \alpha \in S, \beta \in T\;\}$, $S \varepsilon = \{\varepsilon\} \shuffle S = S \shuffle \{ \varepsilon\} = S$, and $S\emptyset=\emptyset S=\emptyset$.
\end{definition}

The following lemma follows directly from the definitions~
and will be used in the proof of Proposition~\ref{prop:pi}.

\begin{lemma}
\label{lemma:eps_shuffle}
  If $\alpha, \beta \in \tsh$, then $\varepsilon(\beta) \cdot \lang(\alpha) \subseteq \lang(\alpha \shuffle \beta)$.
\end{lemma}

\begin{proposition}
\label{prop:pi}
  If $\alpha \in \tsh$, then the set $\pi(\alpha)$
is a support of $\alpha$.
\end{proposition}
\begin{proof}
  We proceed by induction on the structure of $\alpha$. Excluding the
  case where $\alpha$ is $\alpha_0 \shuffle \beta_0$, the proof can be
  found
  in~\cite{b.g.mirkin66:_algor_for_const_base_in,champarnaud01:_from_mirkin_prebas_to_antim}.
  We now describe how to obtain a system of equations corresponding to
  an expression $\alpha_0 \shuffle \beta_0$ from systems for $\alpha_0$
  and $\beta_0$.  Suppose that 
$\pi(\alpha_0) =\{\alpha_1,\ldots,\alpha_n\}$ is a support of $\alpha_0$ and 
$\pi(\beta_0) =  \{\beta_1,\ldots,\beta_m\}$ is a support of $\beta_0$. For $\alpha_0$ and $\beta_0$ consider $\Cmat$,
  $\Amat_{\alpha_0}$, $\Mmat_{\alpha_0}$,
  $\Emat_{\alpha_0}$ and $\Amat_{\beta_0}$,
  $\Mmat_{\beta_0}$, $\Emat_{\beta_0}$ as above.
We wish to show that
\begin{eqnarray*}\pi(\alpha_0 \shuffle \beta_0 ) &= &\{ \alpha_1 \shuffle \beta_1,
\ldots,\alpha_1 \shuffle \beta_m,\ldots,\alpha_n \shuffle \beta_1,\ldots,\alpha_n
\shuffle \beta_m \}\\
&\cup &\{ \alpha_1 \shuffle \beta_0,
\ldots,\alpha_n \shuffle \beta_0 \} \cup \{ \alpha_0 \shuffle \beta_1,
\ldots,\alpha_0 \shuffle \beta_m \}
\end{eqnarray*}
is a support of $\alpha_0 \shuffle \beta_0$. Let
$\Amat_{\alpha_0 \shuffle \beta_0 }$ be the $(n+1)(m+1)$-entry row-matrix whose entires
are
\begin{eqnarray*}
 \begin{bmatrix} \alpha_0 \shuffle \beta_0 \ &\alpha_1 \shuffle \beta_1&
\cdots
& \alpha_n \shuffle \beta_m &\alpha_1\shuffle \beta_0&\cdots &\alpha_n\shuffle \beta_0 &\alpha_0 \shuffle \beta_1&\cdots &\alpha_0 \shuffle \beta_m\end{bmatrix}.
\end{eqnarray*}
Then, $\Emat_{\alpha_0 \shuffle \beta_0}$ is defined as usual, i.e.~containing the values of $\varepsilon(\alpha)$ for all entries $\alpha$ in $\Amat_{\alpha_0 \shuffle \beta_0 }$.

Finally, let $ \Mmat_{\alpha_0\shuffle \beta_0}$ be the $k \times (n+1)(m+1)$ matrix whose
entries $\gamma_{l,(i,j)}$, for $l \in [1,k]$ and $(i,j) \in [0,n] \times [0,m]$, are defined by
\begin{eqnarray*}
  \gamma_{l,(i,j)} \eqsynt \alpha_{li}{\shuffle} \beta_{j}  + 
\alpha_i {\shuffle} \beta_{lj} .
\end{eqnarray*}

Note that, since by the induction hypothesis each $\alpha_{li}$ is a sum of elements in $\pi(\alpha)$  and each $\beta_{lj}$ is a sum of elements in $\pi(\beta)$, after applying distributivity of $\shuffle$ over $+$ each element of $\Mmat_{\alpha_0\shuffle \beta_0}$ is in fact a sum of elements in $\pi(\alpha_0 \shuffle \beta_0)$.  We will show that $\Amat_{\alpha_0 \shuffle\beta_0 }\eqsem \Cmat \cdot \Mmat_{\alpha_0\shuffle\beta_0} + \Emat_{\alpha_0 \shuffle \beta_0}$. For this, consider $\alpha_i \shuffle \beta_j$ for some  $(i,j) \in [0,n] \times [0,m]$. We have 
$ \alpha_i \eqsem a_1\alpha_{1i} + \cdots + a_k  \alpha_{ki} + \varepsilon(\alpha_i)$ and
 $\beta_j \eqsem a_1  \beta_{1j} + \cdots + a_k  \beta_{kj} + \varepsilon(\beta_j).$
Consequently, using properties of $\shuffle$, namely distributivity over $+$, as well as Lemma~\ref{lemma:eps_shuffle}, 

\begin{eqnarray*}
\alpha_i \shuffle \beta_j &\eqsem& (a_1 \alpha_{1i} + \cdots + a_k 
  \alpha_{ki} + \varepsilon(\alpha_i)) \shuffle (a_1 \beta_{1j} + \cdots + a_k 
  \beta_{kj} + \varepsilon(\beta_j))\\
&\eqsem& a_1  \left (\alpha_{1i} \shuffle \beta_j + \alpha_i \shuffle \beta_{1j} + \varepsilon(\beta_j) \alpha_{1i}+
\varepsilon(\alpha_i)\beta_{1j}\right ) \,\, +\,\,  \cdots \,\, + \\
&& a_k  \left (\alpha_{ki} \shuffle \beta_j + \alpha_i \shuffle \beta_{kj} + \varepsilon(\beta_j) \alpha_{ki}+
\varepsilon(\alpha_i)\beta_{kj}\right ) 
+ \varepsilon(\alpha_i \shuffle \beta_j) \\
&\eqsem& a_1  \left (\alpha_{1i} \shuffle \beta_j + \alpha_i \shuffle \beta_{1j} 
\right ) \,\, +\,\,  \cdots \,\, + \\
&& a_k  \left (\alpha_{ki} \shuffle \beta_j + \alpha_i \shuffle \beta_{kj} \right ) 
+ \varepsilon(\alpha_i \shuffle \beta_j) \\
&\eqsem& a_1  \gamma_{1,(i,j)} + \cdots +  a_k \gamma_{k,(i,j)} + \varepsilon(\alpha_i \shuffle \beta_j).
  \end{eqnarray*} \qed
\end{proof}

It is clear from its definition that $\pi(\alpha)$ is finite. In the following proposition, an upper bound for the size of $\pi(\alpha)$ is given. Example~\ref{ex:tight} is a witness that this upper bound is tight.

\begin{proposition} \label{proposition:pdfinite}
Given $\alpha \in \tsh$, one has $|\pi(\alpha)| \leq 2^{|\alpha|_\Sigma}-1$, where $|\alpha|_\Sigma$ denotes the number of alphabet symbols in $\alpha$.
\end{proposition}
\begin{proof}

The proof proceeds by induction on the structure of $\alpha$. It is clear that the result holds for $\alpha=\emptyset$, $\alpha=\varepsilon$ and for $\alpha= a \in \Ash$. Now, suppose the claim is true for $\alpha$ and $\beta$. There are four induction cases to consider. We will make use of the fact that, for $m,n \geq 0$ one has $2^m+2^n-2 \leq 2^{m+n}-1$. For $\alpha^\star$, one has $|\pi(\alpha^\star)|  = |\pi(\alpha) \alpha^\star| = |\pi(\alpha)| \leq 2^{|\alpha|_\Sigma} -1 =  2^{|\alpha^\star|_\Sigma}-1$.
For $\alpha + \beta$, one has $|\pi(\alpha + \beta)|  = |\pi(\alpha) \cup \pi(\beta)| \leq 2^{|\alpha|_\Sigma} -1 + 2^{|\beta|_\Sigma} -1 \leq  2^{|\alpha|_\Sigma + |\beta|_\Sigma} -1 =  2^{|\alpha +\beta|_\Sigma} -1$.
For $\alpha \beta$, one has $|\pi(\alpha  \beta)|  = |\pi(\alpha)\beta \cup \pi(\beta)| \leq 2^{|\alpha|_\Sigma} -1 + 2^{|\beta|_\Sigma} -1 \leq  2^{|\alpha \beta|_\Sigma} -1$.
Finally, for $\alpha \shuffle \beta$, one has $|\pi(\alpha \shuffle \beta)|  
= |\pi(\alpha)\shuffle \pi(\beta) \cup \pi(\alpha) \shuffle \{ \beta \} \cup \{ \alpha \} \shuffle \pi(\beta)| \leq  (2^{|\alpha|_\Sigma}-1)(2^{|\beta|_\Sigma}-1) + 2^{|\alpha|_\Sigma} -1+ 2^{|\beta|_\Sigma}-1 =  2^{|\alpha|_\Sigma +|\beta|_\Sigma} -1 = 2^{|\alpha \shuffle \beta|_\Sigma} -1$.
\qed

\end{proof}

\begin{example}
\label{ex:tight}
Considering again $\alpha_n = a_1 \shuffle \cdots \shuffle a_n$, where
$n \geq 1$, $a_i \not= a_j$ for $1 \leq i \not= j \leq n$, one has
$$|\pi(\alpha_n)| = |\{\; \underset{i \in I}{\shuffle} a_i \mid I
\subsetneq \{1,\ldots, n\} \;\}| = 2^n-1,$$
where we consider $ \underset{i \in \emptyset}{\shuffle} a_i = 1$.
\end{example}

The proof of Proposition~\ref{prop:pi} gives a way to construct a
system of equations for an expression $\tau \in \tsh$, corresponding
to an \NFA that accepts the language represented by $\tau$.  This is
done by recursively computing $\pi(\tau)$ and the matrices
$\Amat_\tau$ and $\Emat_\tau$, obtaining the whole \NFA in the final
step.

In the next section we will show how to build the same \NFA in a more
efficient way using the notion of partial derivative.

\section{Partial Derivatives}
\label{sec:pd}
Recall that the \emph{left-quotient} of a language $L$ w.r.t.~a symbol
$a \in \Ash$ is $$a^{-1}L = \{\;x \mid ax \in L\;\}.$$ The left
quotient of $L$ w.r.t.~a word $x \in \Ash^\star$ is then inductively
defined by $\varepsilon^{-1}L=L$ and $(xa)^{-1}L =
a^{-1}(x^{-1}L)$. Note that for $L_1, L_2 \subseteq \Sigma^\star$ and
$a,b \in \Sigma$ the shuffle operation satisfies $a^{-1}(L_1 \shuffle
L_2) = (a^{-1} L_1) \shuffle L_2 \;\cup \; L_1 \shuffle (a^{-1} L_2).$

\begin{definition}
\label{def:pd-shuffle}
The set of partial derivatives of a term $\tau \in \tsh$ w.r.t.~a
letter $a \in \Ash$, denoted by $\partial_{a}(\tau)$, is inductively
defined by

\begin{equation*}
  \begin{array}{rcl}
   \partial_ {a}(\emptyset)&=&\partial_  {a}(\varepsilon)  = \emptyset \\
    \partial_  {a}(b) &=& \begin{cases} \{ \varepsilon\} \text { if }  b =a \\ \emptyset \text{ otherwise }  \end{cases} 
    \end{array}
    \qquad\qquad
    \begin{array}{rcl}
        \partial_{a}(\alpha^*) &=& \partial_{a}(\alpha) \nodot \alpha^* \\
   \partial_ { a}(\alpha+\beta) &=& \partial_  {a}(\alpha) \, \cup \, \partial_  {a}(\beta) \\
    \partial_  {a}(\alpha \beta) &=& \partial_{a}(\alpha) \nodot \beta\,\, \cup \,\, \varepsilon(\alpha)\nodot  \partial_{a}(\beta)\\
      \partial_{a} (\alpha \shuffle \beta) &=&  \partial_{a} (\alpha) \shuffle \{\beta\} \cup \{\alpha\} \shuffle \partial_{a}(\beta).
    \end{array}
\end{equation*}
The set of partial derivatives of $\tau \in \tsh$ w.r.t.~a word $x \in
\Ash^*$ is inductively defined  by $\partial_\varepsilon(\tau)=\{
\tau\}$ and $ \partial_{x a }(\tau)
= \partial_a(\partial_x(\tau))$, where, given a set $S \subseteq \tsh$, $\partial_a(S)= \bigcup_{\tau \in
  S} \partial_{a}(\tau)$.   
\end{definition}

We denote by $\partial(\tau)$ the set of all partial derivatives of an expression $\tau$, i.e.~$\partial(\tau) = \bigcup_{x \in \Sigma^* } \partial_x(\tau)$, and by $\partial^+(\tau)$ the set of partial derivatives excluding the trivial derivative by $\varepsilon$, i.e.~$\partial^+(\tau) = \bigcup_{x \in \Sigma^+} \partial_x(\tau)$. Given a set $S \subseteq \tsh$, we define $\lang(S) = \bigcup_{\tau \in S} \lang(\tau)$. The following result has a straightforward proof.

\begin{proposition}
\label{prop:shufflequotientproperty}
Given $x \in \Ash^\star$ and $\tau \in \tsh$, one has
 $\lang(\partial_x(\tau)) = x^{-1}\lang(\tau)$.
\end{proposition}

The following properties of  $\partial^+(\tau)$ will be used in the proof of Proposition~\ref{prop:pd_recursion}.

\begin{lemma}
\label{lemma:pdcontainseps}
For $\tau \in \tsh$, the following hold.
\begin{enumerate}
\item If $\partial^+(\tau) \not= \emptyset$, then there is $\alpha_0 \in \partial^+(\tau)$ with $\varepsilon(\alpha_0) = \varepsilon$.
\item If $\partial^+(\tau) = \emptyset$ and $\tau \not=\emptyset$, then $\lang(\tau) = \{ \varepsilon \}$ and $\varepsilon(\tau)=\varepsilon$.
\end{enumerate}
\end{lemma}
\begin{proof}
  \begin{enumerate}
  \item From the grammar rule (\ref{gram:shuffle}) it follows that
    $\emptyset$ cannot appear as a subexpression of a larger term.
    Suppose that there is some $\gamma \in \partial^+(\tau)$. We
    conclude, from Definition~\ref{def:pd-shuffle} and from the
    previous remark, that there is some word $x \in \Sigma^+$ such
    that $x \in \lang(\gamma)$. This is equivalent to $\varepsilon \in
    \lang(\partial_x(\gamma))$, which means that there is some
    $\alpha_0 \in \partial_x(\gamma) \subseteq \partial^+(\tau)$ such
    that $\varepsilon(\alpha_0)=\varepsilon.$
  \item $\partial^+(\tau) = \emptyset$ implies that $\partial_x(\tau)
    = \emptyset$ for all $x \in \Sigma^+$. Thus,
    $\lang(\partial_x(\tau)) = \{ \; y \mid xy \in \lang(\tau) \;\} =
    \emptyset$, and consequently there is no word $z \in \Sigma^+$ in
    $\lang(\tau)$. On the other hand, since $\emptyset$ does not
    appear in $\tau$, it follows that $\lang(\tau) \not=
    \emptyset$. Thus, $\lang(\tau) = \{ \varepsilon \}$.\qed
\end{enumerate}
\end{proof}

\begin{proposition}
  \label{prop:pd_recursion}
  $\partial^+$ satisfies the following,\\
\begin{equation*}
  \begin{array}[t]{rcl}
    \partial^+(\emptyset)&=&\partial^+(\varepsilon) = \emptyset \\
    \partial^+(a) &=& \{ \varepsilon\} \quad (a \in \Ash)\\
     \partial^+(\alpha^*) &=& \partial^+(\alpha)\nodot  \alpha^*
      \end{array}
      \qquad    \begin{array}[t]{rcl}
      \partial^+(\alpha+\beta) &=& \partial^+(\alpha) \cup \partial^+(\beta) \\
      \partial^+(\alpha \beta) &=& \partial^+(\alpha) \nodot \beta
      \cup \partial^+(\beta) \\
      \partial^+(\alpha \shuffle \beta) &=& \partial^+(\alpha)
      \shuffle \partial^+(\beta) \\ 
      &\cup&
      \partial^+(\alpha) \shuffle \{ \beta \} \cup \{ \alpha \}
      \shuffle \partial^+(\beta). 
    \end{array}
  \end{equation*}
\end{proposition}
\begin{proof}
  The proof proceeds by induction on the structure of $\alpha$. It is
  clear that $\partial^+(\emptyset)=\emptyset$,
  $\partial^+(\varepsilon)=\emptyset$
  and, for $a \in \Ash$, $\partial^+(a) = \{\varepsilon\}$.
   
  In the remaining cases, to prove that an inclusion
  $\partial^+(\gamma) \subseteq E$ holds for some expression $E$, we
  show by induction on the length of $x$ that for every $x \in
  \Sigma^+$ one has $\partial_x(\gamma) \subseteq E$.  We will
  therefore just indicate the corresponding computations for
  $\partial_a(\gamma)$ and $\partial_{xa}(\gamma)$, for $a \in
  \Sigma$.  We also make use of the fact that, for any expression
  $\gamma$ and letter $a \in \Sigma$, the set $\partial^+(\gamma)$ is
  closed for taking derivatives w.r.t.~$a$, i.e.,
  $\partial_a(\partial^+(\gamma)) \subseteq \partial^+(\gamma)$.

  Now, suppose the claim is true for $\alpha$ and $\beta$.  There are
  four induction cases to consider.
  \begin{itemize}
  \item For $\alpha + \beta$, we have $\partial_a(\alpha+\beta)
    = \partial_a(\alpha)+\partial_a(\beta) \
    \subseteq \partial^+(\alpha) \cup \partial^+(\beta)$, as well as
    $\partial_{xa}(\alpha+\beta)
    = \partial_a(\partial_x(\alpha+\beta))
    \subseteq \partial_a(\partial^+(\alpha) \cup \partial^+(\beta))
    \subseteq \partial_a(\partial^+(\alpha))
    \cup \partial_a(\partial^+(\beta)) \subseteq \partial^+(\alpha)
    \cup \partial^+(\beta)$. Similarly, one proves that
    $\partial_x(\alpha) \in \partial^+(\alpha+\beta)$ and
    $\partial_x(\beta) \in \partial^+(\alpha+\beta)$, for all $x \in
    \Sigma^+$.

  \item For $\alpha^\star$, we have
    $\partial_a(\alpha^*)=\partial_a(\alpha) \nodot \alpha^*
    \subseteq \partial^+(\alpha) \nodot \alpha^*$, as well as
\begin{align*}
  \partial_{x a} (\alpha^*)=&\ \partial_a(\partial_x(\alpha^*))
  \subseteq \partial_a(\partial^+(\alpha) \nodot \alpha^*)
  \subseteq \partial_a(\partial^+(\alpha)) \nodot \alpha^*
  \cup \partial_a(\alpha^*) \\ &\subseteq \partial^+(\alpha) \nodot
  \alpha^* \cup \partial_a(\alpha) \nodot \alpha^*
  \subseteq \partial^+(\alpha) \nodot \alpha^*.
\end{align*}
Furthermore, $\partial_a(\alpha)\alpha^\star
= \partial_a(\alpha^\star) \subseteq \partial^+(\alpha^\star)$ and
$\partial_{xa}(\alpha) \alpha^\star
= \partial_a(\partial_x(\alpha))\alpha^\star
\subseteq \partial_a(\partial_x(\alpha)\alpha^\star)
\subseteq \partial_a(\partial^+(\alpha^\star))
\subseteq \partial^+(\alpha^\star)$.

\item For $\alpha \beta$, we have $\partial_a(\alpha \beta)
  = \partial_a(\alpha)\beta \cup \varepsilon(\alpha) \partial_a(\beta)
  \subseteq \partial^+(\alpha)\beta \cup \partial^+(\beta)$ and
\begin{align*}
\partial_{xa}(\alpha\beta) = &\ \partial_a(\partial_x(\alpha\beta)) \subseteq
\partial_a(\partial^+(\alpha)\beta \cup \partial^+(\beta)) =
\partial_a(\partial^+(\alpha)\beta)
\cup \partial_a(\partial^+(\beta))\\&
\subseteq \partial_a(\partial^+(\alpha))\beta \cup \partial_a(\beta)
\cup \partial_a(\partial^+(\beta)) \subseteq \partial^+(\alpha)\beta
\cup \partial^+(\beta).
\end{align*} 

Also, $\partial_a(\alpha)\beta \subseteq \partial_a(\alpha\beta)
\subseteq \partial^+(\alpha\beta)$ and
\begin{align*}
\partial_{xa}(\alpha)\beta = \partial_a(\partial_x(\alpha))\beta 
\subseteq
\partial_a(\partial_x(\alpha)\beta) \subseteq 
\partial_a(\partial^+(\alpha\beta)) \subseteq \partial^+(\alpha\beta).
\end{align*}

Finally, if $\varepsilon(\alpha)=\varepsilon$, then $\partial_a(\beta)
\subseteq \partial_a(\alpha \beta)$ and $\partial_{xa}(\beta)
= \partial_a(\partial_x(\beta)) \subseteq
\partial_a(\partial_x(\alpha\beta))=\partial_{xa}(\alpha\beta)$. We
conclude that $\partial_x(\beta) \subseteq \partial_x(\alpha \beta)$
for all $x \in \Sigma^+$, and therefore $\partial^+(\beta)
\subseteq \partial^+(\alpha \beta)$.  Otherwise,
$\varepsilon(\alpha)=\emptyset$, and it follows from
Lemma~\ref{lemma:pdcontainseps} that $\partial^+(\alpha) \not=
\emptyset$, and that there is some $\alpha_0 \in \partial^+(\alpha)$
with $\varepsilon(\alpha_0) = \emptyset$. As above, this implies that
$\partial_x (\beta) \subseteq \partial_x(\alpha_0 \beta)$ for all $x
\in \Sigma^+$.  On the other hand, have already shown that
$\partial^+(\alpha) \beta \subseteq \partial^+(\alpha \beta)$. In
particular, $\alpha_0 \beta \in \partial^+(\alpha \beta)$. From these
two facts, we conclude that $\partial_x(\beta)
\subseteq \partial_x(\alpha_0 \beta)
\subseteq \partial_x(\partial^+(\alpha \beta))
\subseteq \partial^+(\alpha \beta)$, which finishes the proof for the
case of concatenation.
\item For $\alpha \shuffle \beta$, we have 
\begin{align*}
  \partial_a(\alpha \shuffle \beta) = &\ \partial_{a}(\alpha) \shuffle
  \{\beta\} \cup \{\alpha\} \shuffle \partial_a(\beta) \\ &
  \subseteq \partial^+(\alpha) \shuffle \partial^+(\beta) \cup
  \partial^+(\alpha) \shuffle \{ \beta \} \cup \{ \alpha \}
  \shuffle \partial^+(\beta)
\end{align*}   
and
  \begin{align*}
    \partial_{x a } (\alpha \shuffle \beta) \subseteq &
    \ \partial_a(\partial^+(\alpha)
    \shuffle \partial^+(\beta)\cup \partial^+(\alpha) \shuffle \{
    \beta \} \cup \{\alpha\} \shuffle \partial^+(\beta)) \\ =&
    \ \partial_a(\partial^+(\alpha)
    \shuffle \partial^+(\beta))\cup \partial_a(\partial^+(\alpha)
    \shuffle \{ \beta \}) \cup \partial_a(\{\alpha\}
    \shuffle \partial^+(\beta)) \\ = &\ \partial_a(\partial^+(\alpha))
    \shuffle \partial^+(\beta) \cup \partial^+(\alpha)
    \shuffle \partial_a(\partial^+(\beta))
    \cup \partial_a(\partial^+(\alpha)) \shuffle \{ \beta \} \\&
    \cup \partial^+(\alpha) \shuffle \partial_a(\beta)
    \cup \partial_a(\alpha) \shuffle \partial^+(\beta)
    \cup \{\alpha\} \shuffle \partial_a(\partial^+(\beta))\\
    \subseteq &\
    \partial^+(\alpha)
    \shuffle \partial^+(\beta)\cup \partial^+(\alpha) \shuffle \{
    \beta \} \cup \{\alpha\} \shuffle \partial^+(\beta).
\end{align*}
Now we prove that for all $x \in \Sigma^+$, one has
$\partial_x(\alpha) \shuffle \{\beta\} \subseteq \partial_x(\alpha
\shuffle \beta)$, which implies $\partial^+(\alpha) \shuffle \{\beta\}
\subseteq \partial^+(\alpha \shuffle \beta)$. In fact, we have
$\partial_a(\alpha) \shuffle \{\beta\} \subseteq \partial_a(\alpha
\shuffle \beta)$ and
\begin{align*}
  \partial_{xa}(\alpha) \shuffle \{\beta\} \subseteq
  &\ \partial_a(\partial_x(\alpha)) \shuffle \{\beta\} \\ \subseteq
  &\ \partial_a(\partial_x(\alpha) \shuffle \{\beta\})
  \subseteq \partial_a(\partial_x(\alpha \shuffle \beta)) =
\partial_{xa}(\alpha \shuffle \beta).
\end{align*}
Showing that $\{\alpha\} \shuffle \partial_x(\beta)
\subseteq \partial_x(\alpha \shuffle \beta)$ is analogous.  Finally,
for $x, y \in \Sigma^+$ we have $\partial_x(\alpha)
\shuffle \partial_y(\beta) \subseteq \partial_y(\partial_x(\alpha)
\shuffle \{\beta\} ) \subseteq
\partial_y(\partial_x(\alpha \shuffle \beta )) = \partial_{xy}(\alpha \shuffle \beta) \subseteq \partial^+(\alpha \shuffle \beta)$.\qed
\end{itemize}
\end{proof}

\begin{corollary}
  \label{pd_eq_pi}
Given $\alpha \in \tsh$, one has $\partial^+(\alpha) = \pi(\alpha)$.
\end{corollary}

We conclude that $\partial(\alpha)$
corresponds to the set $\{\alpha\} \cup \pi(\alpha)$, as is the case
for standard regular expressions. It is well known that the set of
partial derivatives of a regular expression gives rise to an
equivalent \NFA, called the Antimirov automaton or partial derivative
automaton, that accepts the language determined by that expression.
This remains valid in our extension of the partial derivatives to
regular expressions with shuffle.
\begin{definition}
Given $\tau \in \tsh$, we define
the partial derivative automaton associated to $\tau$ by
$$\apd(\tau) = \langle \partial(\tau), \Sigma, \{\tau \}, \delta_\tau,
F_\tau\rangle,$$ where
 $F_\tau=\{\; \gamma \in \partial(\tau) \mid \varepsilon(\gamma) =\varepsilon \;\}$ and $\delta_\tau(\gamma, a) = \partial_a(\gamma)$. 
\end{definition}
It is easy to see that the following holds.
\begin{proposition}\label{proposition:equivpdska}
  For every state $\gamma \in \partial(\tau)$, the right language
  $\lang_ \gamma$ of $\gamma$ in $\nfa(\tau)$ is equal to
  $\lang(\gamma)$, the language represented by $\gamma$. In
  particular, the language accepted by $\apd(\tau)$ is exactly
  $\lang(\tau)$.
\end{proposition}

Note that for the \res $\alpha_n$ considered in~examples~\ref{ex:lang}
and \ref{ex:tight}, $\apd(\alpha_n)$ has $2^n$ states which is exactly
the bound presented by Mayer and Stockmeyer.

\section{Average State Complexity of the Partial Derivative Automaton}
\label{sec:average}

In this section, we estimate the asymptotic average size of the number
of states in partial derivative automata.  This is done by the use of
the standard methods of analytic combinatorics as expounded by
Flajolet and Sedgewick~\cite{flajolet08:_analy_combin}, which apply to
generating functions $A(z)=\sum_n a_nz^n$ associated to combinatorial
classes. Given some measure of the objects of a class $\mathcal{A}$,
the coefficient $a_n$ represents the sum of the values of this measure
for all objects of size $n$. We will use the notation $[z^n]A(z)$ for
$a_n$. For an introduction of this approach applied to formal
languages, we refer to Broda \emph{et
  al.}~\cite{broda13:_hitch_guide_descr_compl_analy_combin}. In order
to apply this method, it is necessary to have an unambiguous
description of the objects of the combinatorial class, as is the case
for the specification of $\tsh$-expressions without $\emptyset$ in
(\ref{gram:shuffle}).  For the length or size of an $\tsh$-expression
$\alpha$ we will consider the number of symbols in $\alpha$, not
counting parentheses.  Taking $k= |\Ash|$, we compute from
(\ref{gram:shuffle}) the generating functions $R_k(z)$ and $L_{k}(z)$,
for the number of $\tsh$-expressions without $\emptyset$ and the
number of alphabet symbols in $\tsh$-expressions without $\emptyset$,
respectively. Note that excluding one object, $\emptyset$, of size $1$
has no influence on the asymptotic study.

According to the specification in (\ref{gram:shuffle}) the generating function $R_k(z)$ for the number of $\tsh$-expressions without $\emptyset$ satisfies
\begin{eqnarray*}
  R_k(z) &=& z + kz +3z R_k(z)^2 + z R_k(z),
\end{eqnarray*}
thus,
\begin{eqnarray*}
  R_k(z) = \frac{(1-z) - \sqrt{\Delta_k(z)}}{6z}, \mbox{ where } \Delta_k(z)=1 - 2z -(11+12k)z^2.
\end{eqnarray*}
The radius of convergence of $R_k(z)$ is 
 $ \rho_k = \frac{-1 +2\sqrt{3+3k}}{11+12k}$.
Now, note that the number of letters $l(\alpha)$ in an expression $\alpha$ satisfies: $l(\varepsilon)=0$, in $l(a)=1$, for $a \in \Ash$, $l(\alpha+\beta)=l(\alpha)+l(\beta)$, etc. From this, we conclude that the generating function $L_k(z)$ satisfies
\begin{eqnarray*}
  L_k(z) &=& kz +3z L_k(z) R_k(z) + z L_k(z),
\end{eqnarray*}
thus,
\begin{eqnarray*}
  L_k(z) = \frac{(-kz)}{6zR_k(z)+z-1} = \frac{kz}{\sqrt{\Delta_k(z)}}.
\end{eqnarray*}
Now, let $P_k(z)$ denote the generating function for the size of
$\pi(\alpha)$ for $\tsh$-expressions without $\emptyset$. From
Definition~\ref{def:pi} it follows that, given an expression $\alpha$,
an upper bound, $p(\alpha)$, for the number of elements\footnote{This
  upper bound corresponds to the case where all unions in
  $\pi(\alpha)$ are disjoint.}  in the set $\pi(\alpha)$ satisfies:
\begin{equation*}
\begin{array}{rcl}
 p(\varepsilon)&=&0\\
p(a)&=&1, \ \  \text{for $a \in \Ash$}\\
p(\alpha^\star) & = & p(\alpha)
\end{array}\qquad
\begin{array}{rcl}
p(\alpha+\beta)&=&p(\alpha)+p(\beta)\\
p(\alpha\beta) & = &p(\alpha)+p(\beta)\\
p(\alpha \shuffle \beta) &= &p(\alpha)p(\beta) + p(\alpha) + p(\beta).
\end{array}
\end{equation*}
 From this, we conclude that the generating function $P_k(z)$ satisfies
\begin{eqnarray*}
  P_k(z) &=& kz +6z P_k(z) R_k(z) + z P_k(z) + z P_k(z)^2,
\end{eqnarray*}
thus
\begin{eqnarray*}
  P_k(z) &=& Q_k(z) + S_k(z),
\end{eqnarray*}
where 
\begin{eqnarray*}
  Q_k(z) &=&\frac{\sqrt{\Delta_k(z)}}{z}, \qquad \qquad S_k(z) =  - \frac{\sqrt{\Delta'_k(z)}}{z},
\end{eqnarray*}
and $\Delta'_k(z) = 1-2z -(11+16k)z^2$.  The radii of convergence of
$Q_k(z)$ and $S_k(z)$ are respectively $\rho_k$ (defined above) and
$\rho'_{k} = \frac{-1+2\sqrt{3+4k}}{11+16k}.$

\subsection{Asymptotic analysis}
\label{sec:asymp}
A generating function $f$ can be seen as a complex analytic function,
and the study of its behaviour around its dominant singularity $\rho$
(in case there is only one, as it happens with the functions here
considered) gives us access to the asymptotic form of its
coefficients.  In particular, if $f(z)$ is analytic in some
appropriate neighbourhood of $\rho$, then one has the
following~\cite{flajolet08:_analy_combin,nicaud09:_averag_size_of_glush_autom,broda13:_hitch_guide_descr_compl_analy_combin}:
  \begin{enumerate}
  \item \label{coeff1} if $f(z)=a-b\sqrt{1-z/\rho}+o\left(\sqrt{1-z/\rho}\right)$, with
    $a,b\in \RR$, $b\not=0$, then $$[z^n]f(z)\sim
    \frac{b}{2\sqrt{\pi}}\,\rho^{-n}n^{-3/2};$$
  \item \label{coeff2}
    if $f(z)=\frac{a}{\sqrt{1-z/\rho}}+o\left(\frac{1}{\sqrt{1-z/\rho}}\right)$, 
    with $a\in \RR$, and $a\not=0$, then $$[z^n]f(z)\sim
    \frac{a}{\sqrt{\pi}}\,\rho^{-n}n^{-1/2}.$$
  \end{enumerate}
Hence, by \ref{coeff1}. one has for the number of $\tsh$-expressions of size $n$, 
\begin{eqnarray}
[z^n]R_k(z) &=& \frac{(3+3k)^{\frac14}}{6 \sqrt{\pi}}   \rho_k^{-n-\frac12}(n+1)^{- \frac32}
\end{eqnarray}
and by \ref{coeff2}. for the number of alphabet symbols in all expression of size $n$,
\begin{eqnarray}
[z^n]L_k(z) &=& \frac{k}{2 \sqrt{\pi} (3+3k)^{\frac14} }   \rho_k^{-n+\frac12}n^{- \frac12}.
\end{eqnarray}
Consequently, the average number of letters in an expression of size $n$, which we denote by $avL$, is asymptotically given by
\begin{eqnarray*}
  avL &=& \frac{[z^n]L_k(z)}{[z^n]R_k(z)} = \frac{3k \rho_k}{\sqrt{3+3k}} \frac{(n+1)^{\frac32}}{n^\frac12}.
\end{eqnarray*}
Finally, by \ref{coeff1}., one has for the size of expressions of size $n$, 
\begin{eqnarray*}
  [z^n]P_k(z) &=& [z^n]Q_k(z) + [z^n]S_k(z) \\
              &=& \frac{(3+3k)^\frac14 \rho_{k}^{-n-\frac12} + (3+4k)^\frac14 (\rho' _{k})^{-n-\frac12}}{2 \sqrt{\pi}} (n+1)^{-\frac32},
\end{eqnarray*}
and the average size of $\pi(\alpha)$ for an expression $\alpha$ of
size $n$, denoted by $avP$, is asymptotically given by
\begin{eqnarray*}
  avP  &=& \frac{[z^n]P_k(z)}{[z^n]R_k(z)}.
\end{eqnarray*}
Taking into account Proposition~\ref{proposition:pdfinite}, we want to
compare the values of $\log_2 avP$ and $avL$. In fact, one has
\begin{eqnarray*}
  \lim_{n,k \to \infty} \frac{\log_2 avP}{avL} &=& \log_2 \frac43 \sim 0.415.
\end{eqnarray*}
This means that,
\begin{eqnarray*}
  \lim_{n,k \to \infty} avP^{1/avL} &=& \frac43.
\end{eqnarray*}

Therefore, one has the following significant improvement, when
compared with the worst case, for the average case upper bound.

\begin{proposition}
  For large values of $k$ and $n$ an upper bound for the average number
  of states of $\apd$ is $(\frac43)^{|\alpha|_\Sigma}$. 
\end{proposition}

%
%
\section{Conclusion and Future Work}
\label{sec:conclusion}
We implemented in the \fado system~\cite{fado} the construction of the
$\apd$ for \res with shuffle and performed some experimental tests for
small values of $n$ and $k$. Those experiments over statistically
significant samples of uniform random generated \res suggest that the
upper bound obtained in the last section falls way short of its true
value. This is not surprising as in the construction of
$\pi(\alpha)\cup \{\alpha\}$ repeated elements can occur. 

In previous
work~\cite{broda11:_averag_state_compl_of_partial_deriv_autom}, we
identified classes of standard \res that capture a significant
reduction on the size of $\pi(\alpha)$. In the case of \res with
shuffle, those classes enforce only a marginal reduction in the number
of states, but a drastic increase in the complexity of the associated
generating function. Thus the expected gains don't seem to justify its
quite difficult asymptotic study.

Sulzmann and Thiemann~\cite{sulzmann15:_deriv_for_regul_shuff_expres}
extended the notion of Brzozowski derivative for several variants of
the shuffle operator. It will be interesting to carry out a
descriptional complexity study of those constructions and to see if it
is interesting to extend the notion of partial derivative to those
shuffle variants.

An extension of the partial derivative construction for extended \res
with intersection and negation was recently presented by Caron
\emph{et. al}~\cite{caron11:_partial_deriv_of_exten_regul_expres}.  It
will be also interesting to study the average complexity of this
construction.


\end{document}